\documentclass[10pt]{amsart}
\usepackage{amssymb,MnSymbol}
\usepackage{amsthm,amsmath}

\title{Locally monotone Boolean and pseudo-Boolean functions}

\author{Miguel Couceiro}
\address{Mathematics Research Unit, FSTC, University of Luxembourg, 6, rue Coudenhove-Kalergi, L-1359 Luxembourg, Luxembourg} \email{miguel.couceiro[at]uni.lu }

\author{Jean-Luc Marichal}
\address{Mathematics Research Unit, FSTC, University of Luxembourg, 6, rue Coudenhove-Kalergi, L-1359 Luxembourg, Luxembourg} \email{jean-luc.marichal[at]uni.lu }

\author{Tam\'as Waldhauser}
\address{Mathematics Research Unit, FSTC, University of Luxembourg, 6, rue Coudenhove-Kalergi, L-1359 Luxembourg, Luxembourg and Bolyai Institute, University of Szeged, Aradi v\'ertan\'uk tere 1, H-6720 Szeged, Hungary}
\email{twaldha[at]math.u-szeged.hu }

\date{March 27, 2012}

\begin{document}

\theoremstyle{plain}
\newtheorem{theorem}{Theorem}%[section]% Supprimer [section] pour une num�rotation lin�aire
\newtheorem{lemma}[theorem]{Lemma}
\newtheorem{proposition}[theorem]{Proposition}
\newtheorem{corollary}[theorem]{Corollary}
\newtheorem{fact}[theorem]{Fact}
\newtheorem*{main}{Main Theorem}
\newtheorem{problem}{Problem}

\theoremstyle{definition}
\newtheorem{definition}[theorem]{Definition}
\newtheorem{example}[theorem]{Example}

\theoremstyle{remark}
\newtheorem*{conjecture}{\indent Conjecture}
\newtheorem{remark}{Remark}
\newtheorem{claim}{Claim}

\newcommand{\N}{\mathbb{N}}
\newcommand{\R}{\mathbb{R}}
\newcommand{\B}{\mathbb{B}}
\newcommand{\bfa}{\mathbf{a}}
\newcommand{\bfx}{\mathbf{x}}
\newcommand{\bfy}{\mathbf{y}}
\newcommand{\bfz}{\mathbf{z}}
\newcommand{\med}{\mathrm{med}}
\newcommand{\sign}{\mathrm{sign}}

\begin{abstract}
We propose local versions of monotonicity for Boolean and pseudo-Boolean functions: say that a pseudo-Boolean (Boolean) function is $p$-locally
monotone if none of its partial derivatives changes in sign on tuples which differ in less than $p$ positions. As it turns out, this
parameterized notion provides a hierarchy of monotonicities for pseudo-Boolean (Boolean) functions.

Local monotonicities are shown to be tightly related to lattice counterparts of classical partial derivatives via the notion of permutable
derivatives. More precisely, $p$-locally monotone functions are shown to have $p$-permutable lattice derivatives and, in the case of symmetric
functions, these two notions coincide. We provide further results relating these two notions, and present a classification of $p$-locally
monotone functions, as well as of functions having $p$-permutable derivatives, in terms of certain forbidden ``sections'', i.e., functions which
can be obtained by substituting constants for variables. This description is made explicit in the special case when $p=2$.
\end{abstract}



\maketitle

%---------------------------------------------------------------------------------------------- Section 1
\section{Introduction}

Throughout this paper, let $[n]=\{1,\ldots,n\}$ and $\B=\{0,1\}$. We are interested in the so-called Boolean functions $f\colon\B^n\to\B$ and
pseudo-Boolean functions $f\colon\B^n\to\R$, where $n$ denotes the arity of $f$. The pointwise ordering of functions is denoted by $\leq $,
i.e., $f\leq g$ means that $f(\bfx)\leq g(\bfx)$ for all $\bfx\in\B^n$. The negation of $x\in\B$ is defined by $\overline{x}=x\oplus 1$, where
$\oplus$ stands for addition modulo $2$. For $x,y\in\B$, we set $x\wedge y=\min(x,y)$ and $x\vee y=\max(x,y)$.

For $k\in [n]$, $\bfx\in\B^n$, and $a\in\B$, let $\bfx_{k}^{a}$ be the tuple in $\B^{n}$ whose $i$-th component is $a$, if $i=k$, and $x_i$,
otherwise. We use the shorthand notation $\bfx_{jk}^{ab}$ for $(\bfx_j^a)_k^b=(\bfx_k^b)_j^a$. More generally, for $S\subseteq [n]$,
$\bfa\in\B^n$, and $\bfx\in\B^S$, let $\bfa_S^{\bfx}$ be the tuple in $\B^n$ whose $i$-th component is $x_i$, if $i\in S$, and $a_i$, otherwise.

Let $i\in [n]$ and $f\colon\B^n\to\R$. A variable $x_i$ is said to be \emph{essential} in $f$, or that $f$ \emph{depends} on $x_i$, if there
exists $\bfa\in\B^n$ such that $f(\bfa^0_i)\neq f(\bfa^1_i)$. Otherwise, $x_i$ is said to be \emph{inessential} in $f$. Let $S\subseteq [n]$ and
$f\colon \B^n\to\R$. We say that $g\colon \mathbb{B}^{S}\rightarrow \mathbb{R}$ is an \emph{$S$-section} of $f$ if there exists $\bfa\in\B^n$ such that
$g(\bfx)=f(\bfa_S^{\bfx})$ for all $\bfx\in\B^S$. By a \emph{section} of $f$ we
mean an $S$-section of $f$ for some $S\subseteq [n]$, i.e., any function which can be obtained from $f$ by replacing some of its variables by
constants.

The \emph{(discrete) partial derivative} of $f\colon\B^n\to\R$ with respect to its $k$-th variable is the function $\Delta _{k}f\colon
\B^n\to\R$ defined by $\Delta_kf(\bfx)=f(\bfx_k^1)-f(\bfx_k^0)$; see \cite{FolHam05,GraMarRou00}. Note that $\Delta_kf$ does not depend on its
$k$-th variable, hence it could be regarded as a function of arity $n-1$, but for notational convenience we define it as an $n$-ary function.

A pseudo-Boolean function $f\colon\B^n\to\R$ can always be represented by a multilinear polynomial of degree at most $n$ (see \cite{HamRud68}),
that is,
\begin{equation}\label{eq:PBF}
f(\bfx) ~=~ \sum_{S\subseteq [n]}a_S\,\prod_{i\in S}x_i\, ,
\end{equation}
where $a_S\in\R$. For instance, the multilinear expression for a binary pseudo-Boolean function is given by
\begin{equation}\label{eq:PBF22}
a_0+a_1\, x_1+a_2\, x_2+a_{12}\, x_1x_2\, .
\end{equation}
This representation is very convenient for computing the partial derivatives of $f$. Indeed, $\Delta_{k}f$ can be obtained by applying the
corresponding formal derivative to the multilinear representation of $f$. Thus, from (\ref{eq:PBF}), we immediately obtain
\begin{equation}\label{eq:PBFdiff}
\Delta_{k}f(\bfx) ~=~ \sum_{S\ni k}a_S\,\prod_{i\in S\setminus\{k\}}x_i\, .
\end{equation}

We say that $f$ is \emph{isotone} (resp.\ \emph{antitone}) \emph{in its $k$-th variable} if $\Delta_kf(\bfx)\geq 0$ (resp.\ $\Delta_kf(\bfx)\leq
0$) for all $\bfx\in\B^n$. If $f$ is either isotone or antitone in its $k$-th variable, then we say that $f$ is \emph{monotone in its $k$-th
variable}. If $f$ is isotone (resp.\ antitone, monotone) in all of its variables, then $f$ is an \emph{isotone} (resp.\ \emph{antitone},
\emph{monotone}) \emph{function}.\footnote{Note that the terms ``positive'' and ``nondecreasing'' (resp.\ ``negative'' and ``nonincreasing'')
are often used instead of isotone (resp.\ antitone), and it is also customary to use the word ``monotone'' only for isotone functions.} It is
clear that any section of an isotone (resp.\ antitone, monotone) function is also isotone (resp.\ antitone, monotone).
Thus defined, a function $f\colon\mathbb{B}^{n}\rightarrow \mathbb{R}$ is monotone if and only if none of its partial derivatives changes in sign on $\B^n$.

Noteworthy examples of monotone functions include the so-called \emph{pseudo-polyno{\-}mial functions} \cite{CW,CW1} which play an important role, for instance, in the qualitative approach to decision making; for general background see, e.g., \cite{BouDubPraPir09,DPS}. In the current setting, pseudo-polynomial functions can be thought of as compositions $p\circ(\varphi_1,\ldots,\varphi_n)$ of (lattice) polynomial functions $p\colon [a,b]^n\to [a,b]$, $a<b$, with unary functions $\varphi_i\colon \B \to [a,b]$, $i\in [n]$. Interestingly, pseudo-polynomial functions $f\colon \B^n\to\R$ coincide exactly with those pseudo-Boolean functions that are monotone.

\begin{theorem}
A pseudo-Boolean function is monotone if and only if it is a pseudo-polynomial function.
\end{theorem}

\begin{proof}
Clearly, every pseudo-polynomial function is monotone. For the converse, suppose that $f\colon \B^n\to \R$ is monotone and let $a\in\R$ be the minimum and $b\in\R$ the maximum of $f$. Constant functions are obviously pseudo-polynomial functions, therefore we assume $a<b$. Define $\varphi_i\colon\B\to\{a,b\}$ by $\varphi_i(0)=a$ and $\varphi_i(1)=b$ if $f$ is isotone in its $i$-th variable and $\varphi_i(0)=b$ and $\varphi_i(1)=a$ otherwise. Let $p\colon\{a,b\}^n\to [a,b]$ be given by $p=f\circ(\varphi_1^{-1},\ldots,\varphi_n^{-1})$. Thus defined, $p$ is isotone (i.e., order-preserving) in each variable and hence, by Theorem D in Goodstein \cite[p.~237]{Goo67}, there exists a polynomial function $p'\colon [a,b]^n\to [a,b]$ such that $p'|_{\{a,b\}^n}=p$. Therefore $f$ is the pseudo-polynomial function $p'\circ(\varphi_1,\ldots,\varphi_n)$.
\end{proof}

%Interestingly, in the special case of Boolean functions,  pseudo-polynomial functions  $f\colon \B^n\to \B$ coincide exactly with those functions that are monotone. Indeed, each pseudo-polynomial function is clearly monotone. For the converse, let $f\colon \B^n\to \B$ be monotone and let $p\colon \B^n\to \B$ be given by $f(\varphi_1(x_1),\ldots, \varphi_n(x_n))$ where $\varphi_i=\mathrm{id}$ if $\Delta_if(\bfx)\geq0$, and $\varphi_i=\neg$ otherwise. Then $p$ is isotone, and since every isotone Boolean function is a polynomial function, the claim holds. (In fact, we will show at the end of Section 2 that this claim also holds for pseudo-Boolean functions.)

In the special case of Boolean functions, monotone functions are most frequent among functions of small (essential) arity. For instance, among binary functions $f\colon \B^2\to \B$, there are exactly two non-monotone functions, namely the Boolean sum $x_1\oplus x_2$ and its negation $x_1\oplus x_2\oplus 1$. Each of these functions is in fact highly non-monotone in the sense that any of its partial derivatives changes in sign when negating its unique essential variable; this is not the case, e.g., with $f(x_{1},x_{2},x_{3})=x_{1}-x_{1}x_{2}+x_{2}x_{3}$ which is non-monotone but none of its partial derivatives changes in sign when negating any of its variables (see Example~\ref{non-mono} below).

% (x_{1}\wedge \overline{x}_{2})\vee(x_{2}\wedge x_3)

This fact motivates the study of these ``skew"  functions, i.e., these highly non-monotone functions. To formalize this problem we propose the following parameterized relaxations of monotonicity: a function $f\colon\B^n\to\R$ is $p$-locally monotone if none of its partial derivatives changes in sign when negating less than $p$ of its variables, or equivalently, on tuples which differ in less than $p$ positions. With this terminology, our problem reduces to asking which Boolean functions are not $2$-locally monotone. As we will see (Corollary~\ref{cor Boolean 2-local via sections}), these are precisely those functions that have the Boolean sum or its negation as a binary section.

In this paper we extend this study to pseudo-Boolean functions and show that these parameterized relaxations of monotonicity are tightly related to the following lattice versions of partial derivatives. For $f\colon\B^n\to\R$ and $k\in [n]$, let
$\wedge_kf\colon\B^n\to\R$ and $\vee_kf\colon\B^n\to\R$ be the \emph{partial lattice derivatives} defined by
$$
\wedge_kf(\bfx) ~=~ f(\bfx_k^0)\wedge f(\bfx_k^1)\qquad \mbox{and}\qquad \vee_kf(\bfx) ~=~ f(\bfx_k^0)\vee f(\bfx_k^1).
$$
The latter, known as the $k$-th \emph{join derivative} of $f$, was proposed by Fadini \cite{Fad61} while the former, known as the $k$-th
\emph{meet derivative} of $f$, was introduced by Thayse \cite{Tha73}. In \cite{Tha78} these lattice derivatives were shown to be related to
so-called prime implicants and implicates of Boolean functions which play an important role in the consensus method for Boolean and
pseudo-Boolean functions. For further background and applications see, e.g., \cite{CraHam11,DavDesTha78,FolHam00,PosSte04,Tha81}.

Observe that, just like in the case of the partial derivative $\Delta_kf$, the $k$-th variable of each of the lattice derivatives $\wedge_kf$
and $\vee_kf$ is inessential.

The following proposition assembles some basic properties of lattice derivatives.

\begin{proposition}\label{prop basic properties of lattice derivatives}
For any pseudo-Boolean functions $f,g\colon\B^n\rightarrow\R$ and $j,k\in [n]$, $j\neq k$, the following hold:
\begin{enumerate}
\item[$(i)$] $\wedge_k\wedge_kf=\wedge_kf$ and $\vee_k\vee_kf=\vee_kf$;

\item[$(ii)$] if $f\leq g$, then $\wedge_kf\leq \wedge_kg$ and $\vee_kf\leq \vee_kg$;

\item[$(iii)$] $\wedge_j\wedge_kf=\wedge_k\wedge_jf$ and $\vee_j\vee_kf=\vee_k\vee_jf$;

\item[$(iv)$] $\vee_{k}\wedge_jf\leq \wedge_j\vee_kf$.
\end{enumerate}
\end{proposition}

From equations (\ref{eq:PBF}) and (\ref{eq:PBFdiff}) it follows that every function is (up to an additive constant) uniquely determined by its
partial derivatives. As it turns out, this does not hold when lattice derivatives are considered. However, as we shall see (Theorem~\ref{thm reconstruction from derivatives}), there are only
two types of such exceptions.

Now, if an $n$-ary pseudo-Boolean function is $2$-locally monotone, then for every $j, k\in [n]$, $j\neq k$, we have
$\vee_{k}\wedge_{j}f=\wedge_{j}\vee_{k}f$ (see Lemma~\ref{lemma local ==> 2-permutable} below). This motivates the notion of permutable lattice
derivatives. As it turns out, $p$-local monotonicity of $f$ implies permutability of $p$ of its lattice derivatives (see Theorem~\ref{thm
p-local ==> p-permutable}). However the converse does not hold (see Example \ref{ex n-permutable but only 3-local}).

The structure of this paper goes as follows. In Section 2 we formalize the notion of $p$-local monotonicity and show that it gives rise to a
hierarchy of monotonicities whose largest member is the class of all $n$-ary pseudo-Boolean functions (this is the case when $p=1$) and whose
smallest member is the class of $n$-ary monotone functions (this is the case when $p=n$). We also provide a characterization of $p$-locally
monotone functions in terms of ``forbidden'' sections; as mentioned, this characterization is made explicit in the special case when $p=2$. In Section 3 we
introduce the notion of permutable lattice derivatives. Similarly to local monotonicity, the notion of permutable lattice derivatives gives rise to nested
classes, each of which is also described in terms of its sections. In the Boolean case and for $p=2$, these two parameterized notions are shown
to coincide; this does not hold for pseudo-Boolean functions even when $p=2$ (see Example~\ref{ex:5dsf67}). (At the end of Section 3 we also provide some game-theoretic interpretations of $p$-local monotonicity and $p$-permutability of lattice derivatives.) However, in Section 4, we show that
a symmetric function is $p$-locally monotone if and only if it has $p$-permutable lattice derivatives. In the last section we discuss directions for future research.

%---------------------------------------------------------------------------------------------- Section 2
\section{Local monotonicities}

The following definition formulates a local version of monotonicity given in terms of Hamming distance between tuples. In what follows we assume that $p\in [n]$.

\begin{definition}
We say that $f\colon\mathbb{B}^{n}\to\mathbb{R}$ is \emph{$p$-locally monotone} if, for every $k\in [n]$ and every $\mathbf{x},\mathbf{y}\in
\mathbb{B}^{n}$, we have
$$%\begin{equation}\label{eq def p-local}
\sum_{i\in [n]\setminus \{k\}}|x_{i}-y_{i}|~<~p\quad \Rightarrow \quad \Delta_{k}f(\mathbf{x})\,\Delta_{k}f(\mathbf{y})~\geq ~0.
$$%\end{equation}
\end{definition}

Any $p$-locally monotone pseudo-Boolean function is also $p'$-locally monotone for every $p'\leq p$. Every function
$f\colon\mathbb{B}^{n}\to\mathbb{R}$ is $1$-locally monotone, and $f$ is $n$-locally monotone if and only if it is monotone. Thus $p$-local
monotonicity is a relaxation of monotonicity, and the nested classes of $p$-locally monotone functions for $p=1,\ldots,n$ provide a hierarchy of
monotonicities for $n$-ary pseudo-Boolean functions. The weakest nontrivial condition is $2$-local monotonicity, therefore we will simply say
that $f$ is \emph{locally monotone} whenever $f$ is $2$-locally monotone.\footnote{In \cite{GotLin94},
local monotonicity is used to refer to Boolean functions which are monotone (i.e., isotone or antitone in each variable).} If $f$ is $p$-locally monotone for some $p<n$ but not $(p+1)$-locally
monotone, then we say that $f$ is \emph{exactly $p$-locally monotone}, or that the \emph{degree of local monotonicity} of $f$ is $p$.

If $f\colon\mathbb{B}^{n}\to\mathbb{B}$ is a Boolean function, then $\Delta_{k}f(\mathbf{x}) \in\{ {-1},0,1\}$ for all $\mathbf{x}\in
\mathbb{B}^{n}$, hence the condition $\Delta_{k}f(\mathbf{x})\,\Delta_{k}f(\mathbf{y})\geq 0$ in the definition of $p$-local monotonicity is
equivalent to
\begin{equation}\label{eq:re7t6}
|\Delta_{k}f(\mathbf{x})-\Delta_{k}f( \mathbf{y})|~\leq ~1.
\end{equation}
From this it follows that a Boolean function $f\colon\mathbb{B}^{n}\to\mathbb{B}$ is locally monotone if and only if
\begin{equation}\label{eq:1-Lipschitz}
\big|\Delta_{k}f(\mathbf{x})-\Delta_{k}f(\mathbf{y})\big|~\leq \sum_{i\in \lbrack n]\setminus \{k\}}|x_{i}-y_{i}|.
\end{equation}
(see \cite[Lemma~5.1]{MarMes04} for a proof of (\ref{eq:1-Lipschitz}) in a slightly more general framework). In a sense, the latter identity
means that $\Delta_{k}f$ is ``$1$-Lipschitz continuous''.

The following proposition is just a reformulation of the definition of $p$-local monotonicity.

\begin{proposition}\label{prop p-local}
A function $f\colon\mathbb{B}^{n}\to\mathbb{R}$ is $p$-locally monotone if and only if, for every $k\in [n]$, $S\subseteq [n]\setminus \{k\}$,
with $|S|=p-1$, and every $\mathbf{a}\in \mathbb{B}^{n}$, $\mathbf{x},\mathbf{y}\in \mathbb{B}^{S}$, we have
\begin{equation}\label{eq:re7t6fsd}
\Delta_{k}f(\mathbf{a}_{S}^{\mathbf{x}})\,\Delta _{k}f(\mathbf{a}_{S}^{\mathbf{y}})~\geq ~0.
\end{equation}
Equivalently, a pseudo-Boolean function is $p$-locally monotone if and only if none of its partial derivatives changes in sign when negating less than $p$ of its variables.
\end{proposition}

As a special case, we have that $f\colon\mathbb{B}^{n}\to\mathbb{R}$ is locally monotone if and only if, for every $j,k\in [n]$, $j\neq k$, and
every $\mathbf{x}\in \mathbb{B}^{n}$, we have
\begin{equation}\label{eq:67dsf5}
\Delta_{k}f(\mathbf{x}_{j}^{0})\,\Delta_{k}f(\mathbf{x}_{j}^{1})~\geq ~0.
\end{equation}
Equivalently, a pseudo-Boolean function is locally monotone if and only if none of its partial derivatives changes in sign when negating any of its variables.

By (\ref{eq:re7t6}) we see that, for Boolean functions $f\colon\mathbb{B}^{n}\to\mathbb{B}$, inequality (\ref{eq:67dsf5}) can be replaced with
$|\Delta _{jk}f(\mathbf{x})|\leq 1$, where $\Delta _{jk}f(\mathbf{x})=\Delta_j\Delta_{k}f(\mathbf{x})=\Delta_k\Delta_{j}f(\mathbf{x})$.

\begin{example}\label{ex:6s7dfa}
As observed, the binary Boolean sum
$$
f_{1}(x_{1},x_{2}) ~=~ x_{1}\oplus x_{2} ~=~ x_1+x_2-2x_1x_2
$$
and the binary Boolean equivalence
$$
f_{2}(x_{1},x_{2}) ~=~ \overline{f}_1(x_{1},x_{2}) ~=~ x_{1} \oplus x_{2}\oplus 1  ~=~ 1-x_1-x_2+2x_1x_2
$$
are not locally monotone. Indeed, we have $|\Delta_{12} f_1(x_1,x_2)|=|\Delta_{12} f_2(x_1,x_2)|=2$.
\end{example}

\begin{example}\label{non-mono}
Consider the ternary Boolean function $f\colon\mathbb{B}^{3}\to\mathbb{B}$ given by
\begin{equation*}
f(x_{1},x_{2},x_{3}) ~=~ x_{1}-x_{1}x_{2}+x_{2}x_{3}.
\end{equation*}
Since $\Delta _{2}f$ may change in sign ($\Delta_{2}f(\mathbf{x})=x_{3}-x_{1}$), the function $f$ is not monotone. However, $f$ is locally
monotone since $|\Delta_{12}f(\bfx)|=1$, $|\Delta_{13}f(\bfx)|=0$, and $|\Delta_{23}f(\bfx)|=1$. Thus $f$ is exactly $2$-locally monotone.
Example~\ref{ex f p-local Vf (p-1)-local} in Section~\ref{sect symmetric} provides, for each $p\geq 2$, examples of exactly $p$-locally monotone
functions.
\end{example}

\begin{fact}\label{fact:sdf65}
A function $f\colon\B^n\to\R$ is $p$-locally monotone if and only if so is $\alpha f+\beta$ for every $\alpha,\beta\in\R$, with $\alpha\neq0$.
The same holds for any function obtained from $f$ by negating some of its variables.
\end{fact}

The next theorem gives a characterization of $p$-locally monotone functions in terms of their sections.

\begin{theorem}\label{thm p-local via sections}
A function $f\colon \mathbb{B}^{n}\rightarrow \mathbb{R}$ is $p$-locally monotone if and only if every $p$-ary section of $f$ is monotone.
\end{theorem}

\begin{proof}
We just need to observe that the inequality (\ref{eq:re7t6fsd}) is equivalent to $\Delta_{k}g(\mathbf{x})\Delta_{k}g(\mathbf{y})\geq 0$, where
$g$ is the $p$-ary section of $f$ defined by $g(\mathbf{x}) =f(\mathbf{a}_{S\cup \{k\}}^{\mathbf{x}})$, where $S$ is a $(p-1)$-subset of $[n]\setminus \{k\}$.
Thus $f$ is $p$-locally monotone if and
only if $\Delta_{k}g(\mathbf{x})\Delta_{k}g(\mathbf{y})\geq 0$ holds for every  $\mathbf{x},\mathbf{y}\in\mathbb{B}^{S\cup \{k\}}$,
and for every ${S\cup \{k\}}$-section $g$ of $f$.
\end{proof}

By combining (\ref{eq:67dsf5}) with Theorem~\ref{thm p-local via sections}, we can easily verify the following corollary.

\begin{corollary}\label{cor:7f5s7}
A function $f\colon \mathbb{B}^{n}\rightarrow \mathbb{R}$ is locally monotone if and only if every binary section (\ref{eq:PBF22}) of $f$
satisfies $a_1(a_1+a_{12})\geq 0$ and $a_2(a_2+a_{12})\geq 0$.
\end{corollary}

Since every binary Boolean function is monotone except for $x\oplus y$ and $x\oplus y\oplus 1$, we also obtain the following corollary.

\begin{corollary}\label{cor Boolean 2-local via sections}
A Boolean function $f\colon \mathbb{B}^{n}\rightarrow \mathbb{B}$ is locally monotone if and only if neither $x\oplus y$ nor $x\oplus y\oplus 1$
is a section of $f$.
\end{corollary}

%---------------------------------------------------------------------------------------------- Section 3
\section{Permutable lattice derivatives}

The aim of this section is to relate commutation of lattice derivatives to $p$-local monotonicity. The starting point is the characterization of
locally monotone Boolean functions given in Theorem~\ref{thm Boolean 2-local iff 2-permutable} below.

\begin{lemma}\label{lemma local ==> 2-permutable}
If $f\colon \mathbb{B}^{n}\to\mathbb{R}$ is locally monotone, then $\vee_{k}\wedge_{j}f=\wedge_{j}\vee_{k}f$ for all $j, k\in [n]$, $j\neq k$.
\end{lemma}

\begin{proof}
Let $f\colon \mathbb{B}^{n}\to\mathbb{R}$ be a locally monotone function, and let $j, k\in [n]$, $j\neq k$. Setting $a=f(\mathbf{x}_{jk}^{00})$,
$b=f(\mathbf{x}_{jk}^{01})$, $c=f(\mathbf{x}_{jk}^{10})$, and $d=f(\mathbf{x}_{jk}^{11})$, the desired equality $\vee_{k}\wedge_{j}f(\mathbf{x})
=\wedge_{j}\vee_{k}f(\mathbf{x})$ takes the form
\begin{equation}\label{eq abcd}
(a\wedge c)\vee (b\wedge d) ~=~ (a\vee b)\wedge (c\vee d).
\end{equation}
Since $f$ is $2$-locally monotone, the binary section $g( u,v)=f(\mathbf{x}_{jk}^{uv})$ is monotone, according to Theorem~\ref{thm p-local via
sections}. If $g$ is isotone in $u$, then $a\leq c$ and $b\leq d$, while if $g $ is antitone in $u$, then $a\geq c$ and $b\geq d$. Similarly, we
have either $a\leq b$ and $c\leq d$ or $a\geq b$ and $c\geq d$, depending on whether $g$ is isotone or antitone in $v$. Thus we need to consider
four cases, and in each one of them it is straightforward to verify (\ref{eq abcd}).
\end{proof}

\begin{theorem}\label{thm Boolean 2-local iff 2-permutable}
A Boolean function $f\colon\mathbb{B}^{n}\to\mathbb{B}$ is locally monotone if and only if $\vee_{k}\wedge_{j}f=\wedge_{j}\vee_{k}f$ holds for
all $j, k\in [n]$, $j\neq k$.
\end{theorem}

\begin{proof}
If $f$ is locally monotone, then $\vee_{k}\wedge_{j}f=\wedge_{j}\vee_{k}f$ by Lemma~\ref{lemma local ==> 2-permutable}. If $f$ is not locally
monotone, then Corollary~\ref{cor Boolean 2-local via sections} implies that there exists $\mathbf{a}\in\mathbb{B}^{n}$ and $j, k\in [n]$,
$j\neq k$, such that the binary section $g(u,v) =f(\mathbf{a}_{jk}^{uv})$ is of the form $g(u,v)=u\oplus v$ or $g(u,v)=u\oplus v\oplus 1$. Then
we have
\begin{eqnarray*}
\vee_{k}\wedge_{j}f(\mathbf{a}) &=&( g( 0,0) \wedge g( 1,0) ) \vee ( g(0,1) \wedge g( 1,1) ) =0, \\
\wedge _{j}\vee_{k}f(\mathbf{a}) &=&( g( 0,0) \vee g( 0,1) ) \wedge ( g( 1,0) \vee g( 1,1) ) =1,
\end{eqnarray*}
showing that $\vee_{k}\wedge_{j}f\neq \wedge_{j}\vee_{k}f$.
\end{proof}

As the next example shows, Theorem~\ref{thm Boolean 2-local iff 2-permutable} is not valid for pseudo-Boolean functions.

\begin{example}\label{ex:5dsf67}
Let $f$ be the binary pseudo-Boolean function defined by $f(0,0) =1$, $f(0,1) =4$, $f(1,0) =2$ and $f(1,1) =3$. Then we have
$\vee_{2}\wedge_{1}f=\wedge_{1}\vee_{2}f=3$ and $\vee_{1}\wedge_{2}f=\wedge_{2}\vee_{1}f=2$. However, $f$ is not locally monotone since
$\Delta_{1}f(\mathbf{x}_{2}^{0})\Delta_{1}f(\mathbf{x}_{2}^{1})=-1$.
\end{example}

Lemma~\ref{lemma local ==> 2-permutable} and Theorem~\ref{thm Boolean 2-local iff 2-permutable} motivate the following notion of permutability of lattice derivatives, and its relation to local monotonicities.

\begin{definition}
We say that a pseudo-Boolean function $f\colon \mathbb{B}^{n}\to\mathbb{R}$ \emph{has $p$-permutable lattice derivatives} if, for every
$p$-subset $\{k_{1},\ldots ,k_{p}\}\subseteq [n]$, every choice of the operators $O_{k_{i}}\in \left\{\wedge_{k_{i}},\vee_{k_{i}}\right\}$
$(i=1,\ldots ,p)$, and every permutation $\pi\in S_{p}$, the following identity holds:
$$
O_{k_{1}}\cdots{\,}O_{k_{p}}f ~=~ O_{k_{\pi(1)}}\cdots{\,}O_{k_{\pi(p)}}f.
$$
If $f\colon \mathbb{B}^{n}\to\mathbb{R}$ has $n$-permutable lattice derivatives, then we simply say that $f$ \emph{has permutable lattice derivatives}.
\end{definition}

Every function $f\colon\B^n\to\R$ has $1$-permutable lattice derivatives. We will see in Theorem~\ref{thm (p+1)-permutable ==> p-permutable}
that if a function $f\colon\B^n\to\R$ has $p$-permutable lattice derivatives, then it also has $p'$-permutable lattice derivatives for every
$p'\leq p$.

\begin{fact}\label{fact:saf6d5s}
A function $f\colon\B^n\to\R$ has $p$-permutable lattice derivatives if and only if so has $\alpha f+\beta$ for every $\alpha,\beta\in\R$, with
$\alpha\neq0$. The same holds for any function obtained from $f$ by negating some of its variables.
\end{fact}

\begin{fact}\label{fact:fd5gfd45}
A function $f\colon\B^n\to\R$ has $p$-permutable lattice derivatives if and only if every $p$-ary section of $f$ has permutable lattice derivatives.
\end{fact}

In the particular case when $p=2$, we have the following description of functions having $2$-permutable lattice derivatives. The proof is a
straightforward verification of cases.

\begin{proposition}\label{prop:7f5s7a}
A function $f\colon \mathbb{B}^{n}\rightarrow \mathbb{R}$ has $2$-permutable lattice derivatives if and only if every binary section
(\ref{eq:PBF22}) of $f$ satisfies $a_1\, a_{12}\geq 0$ or $a_2\, a_{12}\geq 0$ or $|a_{12}|\leq |a_1|\vee |a_2|$.
%$$
%a_1\, a_{12}<0~\;\mbox{and}~\; a_2\, a_{12}<0\quad\Rightarrow\quad |a_{12}|\leq |a_1|\vee |a_2|.
%$$
\end{proposition}

Lemma~\ref{lemma local ==> 2-permutable} shows that the class of $2$-locally monotone pseudo-Boolean functions is a subclass of that of pseudo-Boolean functions which have
$2$-permutable lattice derivatives. Example~\ref{ex:5dsf67} then shows that this inclusion is strict.
%From Corollary~\ref{cor:7f5s7} and Proposition~\ref{prop:7f5s7a}, we obtain the following description of
%pseudo-Boolean functions which have $2$-permutable lattice derivatives but are not $2$-locally monotone.
%\begin{corollary}\label{cor:dfsfs68}
%A function $f\colon \mathbb{B}^{n}\rightarrow \mathbb{R}$ has $2$-permutable lattice derivatives but is not $2$-locally monotone if and only if
%every binary section (\ref{eq:PBF22}) of $f$ satisfies the following two conditions:
%\begin{enumerate}
%\item[$(i)$] $a_1\, a_{12}\geq 0$ or $a_2\, a_{12}\geq 0$ or $|a_{12}|\leq |a_1|\vee |a_2|$, and
%
%\item[$(ii)$] $a_1(a_1+a_{12})< 0$ or $a_2(a_2+a_{12})< 0$.
%\end{enumerate}
%Equivalently,
%\begin{enumerate}
%\item[$(i)$] $a_1\, a_{12}<0$ or $a_2\, a_{12}<0$,
%
%\item[$(ii)$] $|a_{12}|> |a_1|\wedge |a_2|$,
%
%\item[$(iii)$] if $a_1\, a_2>0$, then $|a_{12}|\leq |a_1|\vee |a_2|$,
%
%\item[$(iv)$] if $a_1\, a_{12}\geq 0$, then $|a_{12}|> |a_2|$, and
%
%\item[$(v)$] if $a_2\, a_{12}\geq 0$, then $|a_{12}|> |a_1|$.
%\end{enumerate}
%\end{corollary}
On the other hand, according to Theorem~\ref{thm Boolean 2-local iff 2-permutable}, a Boolean function is $2$-locally monotone if and only if it has
$2$-permutable lattice derivatives. Example~\ref{ex n-permutable but only 3-local} below shows that the analogous equivalence does not hold for
$p>2$. However, $p$-local monotonicity implies $p$-permutability of lattice derivatives of any pseudo-Boolean function (see Theorem~\ref{thm
p-local ==> p-permutable} below). To this extent, let us first study how the degree of local monotonicity is affected by taking lattice
derivatives.

\begin{lemma}\label{lemma monotone derivatives}
If $f\colon \mathbb{B}^{n}\to\mathbb{R}$ is monotone, then $\wedge_{j}f$ and $\vee_{j}f$ are also monotone for all $j\in [n]$.
\end{lemma}

 \begin{proof}
Clearly, if $f$ is monotone, then so are $f^0_j(\mathbf{x})=f(\mathbf{x}^0_j)$ and $f^1_j(\mathbf{x})=f(\mathbf{x}^1_j)$, for all $j\in [n]$.
Moreover, if $f$ is isotone (resp.\ antitone) in $x_k$, then both $f^0_j$ and $f^1_j$ are also isotone (resp.\ antitone) in $x_k$.
Since $\wedge$ and $\vee$ are isotone functions, we have that for every $j\in [n]$, both $\wedge_{j}f(\mathbf{x})=f^0_j(\mathbf{x})\wedge
f^1_j(\mathbf{x})$ and $\vee_{j}f(\mathbf{x})=f^0_j(\mathbf{x})\vee f^1_j(\mathbf{x})$ are monotone.
\end{proof}

\begin{theorem}\label{thm (p-1)-locally monotone derivatives}
If $f\colon \mathbb{B} ^{n}\to\mathbb{R}$ is $p$-locally monotone, then $\wedge_{j}f$ and $\vee_{j}f$ are $(p-1)$-locally monotone
for all $j\in [n]$.
\end{theorem}

\begin{proof}
Suppose that $f\colon\mathbb{B}^{n}\to\mathbb{R}$ is $p$-locally monotone. By Theorem~\ref{thm p-local via sections}, it suffices to show that
all $(p-1)$-ary sections of $\wedge_{j}f$ and $\vee_{j}f$ are monotone. We consider only $\vee_{j}f$, the other case can be dealt with in a
similar way.

Let $h$ be a $(p-1)$-ary section of $\vee_{j}f$ defined by $h(\mathbf{x})=\vee_{j}f(\mathbf{a}_{S}^{\mathbf{x}})$ for all $\mathbf{x}\in
\mathbb{B}^{S}$, where $\mathbf{a}\in \mathbb{B}^{n}$ and $S\subseteq [n]$ is a $(p-1)$-subset. Let $T=S\cup\{j\}$, and let us define
$g\colon\mathbb{B}^{T}\to\mathbb{R}$ by $g(\mathbf{y})=f(\mathbf{a}_{T}^{\mathbf{y}})$ for all $\mathbf{y}\in \mathbb{B}^{T}$. Clearly, $g$ is a
section of $f$, and the arity of $g$ is either $p-1$ or $p$ , depending on whether $j$ belongs to $S$ or not. A simple calculation shows that
$h(\mathbf{y}|_{S})=\vee _{j}g(\mathbf{y})$ for all $\mathbf{y}\in\mathbb{B}^{T}$, where $\mathbf{y}|_{S}$ stands for the restriction of
$\mathbf{y}$ to $S$. This means that if $j\notin S$, then $h$ can be obtained from $\vee_{j}g$ by deleting its inessential $j$-th variable, and
$h=\vee _{j}g$ if $j\in S$. Since $f$ is $p$-locally monotone, $g$ is monotone by Theorem~\ref{thm p-local via sections}, thus we can conclude
with the help of Lemma~\ref{lemma monotone derivatives} that $h$ is monotone as well.
\end{proof}

\begin{corollary}\label{cor (p-l)-locally monotone derivatives}
Let $0\leq \ell <p\leq n$. If $f\colon \mathbb{B}^{n}\rightarrow \mathbb{R}$ is $p$-locally monotone, then, for every $\ell$-subset
$\{k_{1},\ldots ,k_{\ell }\}\subseteq [n]$ and every choice of the operators $O_{k_{i}}\in \{\vee_{k_{i}},\wedge_{k_{i}}\}$ $(i=1,\ldots
,\ell)$, the function $O_{k_{1}}\cdots{\,}O_{k_{\ell}}f$ is $(p-\ell)$-locally monotone. In particular, if $\ell \leq p-2$, then $ O_{k_{1}}\cdots{\,}O_{k_{\ell }}f$ is locally monotone.
\end{corollary}

\begin{remark}
We will see in Example~\ref{ex f p-local Vf (p-1)-local} of Section~\ref{sect symmetric} that Theorem~\ref{thm (p-1)-locally monotone
derivatives} cannot be sharpened, i.e., the lattice derivatives of a $p$-locally monotone function are not necessarily $p$-locally monotone,
not even in the case of Boolean functions.
\end{remark}

With the help of Corollary \ref{cor (p-l)-locally monotone derivatives}
 we can now prove the promised implication between $p$-local monotonicity and $p$-permutability of lattice
derivatives, thus generalizing Lemma~\ref{lemma local ==> 2-permutable}.

\begin{theorem}\label{thm p-local ==> p-permutable}
If $f\colon \mathbb{B}^{n}\rightarrow \mathbb{R}$ is $p$-locally monotone, then it has $p$-permutable lattice derivatives.
\end{theorem}

\begin{proof}
Let $f\colon \mathbb{B}^{n}\rightarrow \mathbb{R}$ be a $p$-locally monotone function, let $\{k_{1},\ldots ,k_{p}\}$ be a $p$-subset of $[n]$,
and let $O_{k_{i}}\in \left\{ \wedge _{k_{i}},\vee _{k_{i}}\right\}$ for $i=1,\ldots ,p$. We need to show that for any permutation $\pi \in
S_{p}$ the following identity holds:
\begin{equation*}
O_{k_{1}}\cdots{\,}O_{k_{p}}f ~=~ O_{k_{\pi (1)}}\cdots{\,}O_{k_{\pi (p)}}f.
\end{equation*}
Since $S_{p}$ is generated by transpositions of the form $(i~i+1)$, it suffices to prove that
\begin{equation*}
O_{k_{1}}\cdots{\,}O_{k_{i-1}}O_{k_{i}}O_{k_{i+1}}O_{k_{i+2}}\cdots{\,}O_{k_{p}}f ~=~ O_{k_{1}}\cdots{\,}O_{k_{i-1}}O_{k_{i+1}}O_{k_{i}}O_{k_{i+2}}\cdots{\,}O_{k_{p}}f,
\end{equation*}
and for this it is sufficient to verify that
\begin{equation}
O_{k_{i}}O_{k_{i+1}}g ~=~ O_{k_{i+1}}O_{k_{i}}g,  \label{eq OOg=OOg}
\end{equation}
where $g$ stands for the function $O_{k_{i+2}}\cdots{\,}O_{k_{p}}f$. From Corollary~\ref{cor (p-l)-locally monotone derivatives} it follows that
$g$ is locally monotone, and then Lemma~\ref{lemma local ==> 2-permutable} proves (\ref{eq OOg=OOg}) if one of $O_{k_{i}},O_{k_{i+1}}$ is a meet
and the other is a join derivative. (If both are meet or both are join, then (\ref{eq OOg=OOg}) is trivial.)
\end{proof}

A natural question regarding lattice derivatives is whether a function can be reconstructed from its derivatives. As the next theorem shows,
the answer is positive for almost all functions.

\begin{theorem}\label{thm reconstruction from derivatives}
Let $f,g\colon \mathbb{B}^{n}\rightarrow \mathbb{R}$ be pseudo-Boolean functions such that for all $k\in [n]$ we have $\vee_{k}f=\vee_{k}g$ and
$\wedge_{k}f=\wedge_{k}g$. Then either $f=g$ or there exists a one-to-one function $\alpha \colon \B\to\mathbb{R}$ such that
$f(\mathbf{x})=\alpha (x_{1}\oplus \cdots \oplus x_{n})$ and $g(\mathbf{x})=\alpha (x_{1}\oplus \cdots \oplus x_{n}\oplus 1)$ for all
$\mathbf{x}\in \mathbb{B}^{n}$.
\end{theorem}

\begin{proof}
To make the proof more vivid, we present it through the analysis of the following game. Alice picks a secret function $f\colon
\mathbb{B}^{n}\rightarrow \mathbb{R}$, and Bob tries to identify this function by asking the values of its lattice derivatives. If he can do
this, then he wins, otherwise Alice is the winner. We show that Bob has a winning strategy unless $f$ is a function of the special form in the
statement of the theorem.

Let us regard $\mathbb{B}^{n}$ as the set of vertices of the $n$-dimensional cube, and let Alice write the values of $f$ to the corresponding
vertices. Now the possible winning strategy for Bob is based on the following four basic observations.

\begin{enumerate}
\item[$1.$] \emph{Bob can determine the unordered pair of numbers written to the endpoints of any edge of the cube.} Indeed, the endpoints of an
edge are of the form $\mathbf{x}_{{k}}^{0},\mathbf{x}_{{k}}^{1}$, and it is clear that $\left\{ \wedge _{{k}}f(\mathbf{x}),\vee
_{{k}}f(\mathbf{x})\right\} = \{f(\mathbf{x}_{{k}}^{0}),f(\mathbf{x}_{{k}}^{1})\}$.

\item[$2.$] \emph{If Bob can find the value of }$f$\emph{\ at one point, then he can win.} According to the previous observation, knowing the
value at one vertex of the cube, Bob can figure out the values written to the neighboring vertices. Since the graph of the cube is connected, he
can determine all values of $f$ this way.

\item[$3.$] \emph{If }$f(\mathbf{x}_{\boldsymbol{k}}^{0})=f(\mathbf{x}_{{k}}^{1})$ \emph{for some }$\mathbf{x}\in \mathbb{B}^{n}$, ${k}\in [n]$
\emph{, then Bob can win.} This follows immediately from the first two observations.

\item[$4.$] \emph{If the range of }$f$\emph{\ contains at least three elements, then Bob can win.} We can suppose that the previous observation
does not apply, i.e., for every edge Bob detects a two-element set. If $f$ takes on at least three different values, then, by the connectedness
of the cube, there exists a vertex $\mathbf{x}$ and two edges incident with this vertex such that the two-element sets $E_{1}$ and $E_{2}$
corresponding to these edges are different. Then $E_{1}\cap E_{2}$ must be a one-element set\footnote{If $E_{1}\cap E_{2}$ is empty, then Alice
is cheating!} containing the value of $f( \mathbf{x}) $, and then Bob can win as explained in the second observation.
\end{enumerate}

From these observations we can conclude that Bob has a winning strategy unless the range of $f$ contains exactly two numbers and
$f(\mathbf{x}_{{k}}^{0})\neq f(\mathbf{x}_{{k}}^{1})$, for all $\mathbf{x}\in \mathbb{B}^{n},{k}\in [n]$. This means that $f$ is of the
following form for some $u\neq v\in \mathbb{R}$:
\begin{equation*}
f(\mathbf{x}) ~=~ \left\{ \!\!
\begin{array}{ll}
u{\,}, & \text{if }\left\vert \mathbf{x}\right\vert \text{ is even{\,};} \\
v{\,}, & \text{if }\left\vert \mathbf{x}\right\vert \text{ is odd{\,},}
\end{array}
\right.
\end{equation*}
where $|\bfx|=\sum_{i=1}^nx_i$. In other words, $f(\mathbf{x})=\alpha (x_{1}\oplus \cdots \oplus x_{n})$, where $\alpha (0)=u,\alpha (1)=v$. In
this case Bob can determine $f$ only
up to interchanging $u$ and $v$, i.e., he cannot distinguish $f$ from $g(%
\mathbf{x})=\alpha (x_{1}\oplus \cdots \oplus x_{n}\oplus 1)$, so he has only $50\%$ chance to win. (Indeed, $f$ and $g$ have the same lattice
derivatives, namely their meet derivatives are all constant $u\wedge v $, while their join derivatives are all constant $u\vee v$.)
\end{proof}

The following theorem shows that, as in the case of local monotonicity, the classes of functions having permutable lattice derivatives form a
chain under inclusion.

\begin{theorem}\label{thm (p+1)-permutable ==> p-permutable}
If $f\colon \mathbb{B}^{n}\rightarrow \mathbb{R}$ has $(p+1)$-permutable lattice derivatives, then $f$ has $p$-permutable lattice derivatives.
\end{theorem}

\begin{proof}
Let $f\colon \mathbb{B}^{n}\rightarrow \mathbb{R}$ be a function that has $(p+1)$-permutable lattice derivatives. Using the same notation as in
the proof of Theorem~\ref{thm p-local ==> p-permutable}, it suffices to prove that
$$
O_{k_{1}}\cdots{\,}O_{k_{i-1}}O_{k_{i}}O_{k_{i+1}}O_{k_{i+2}}\cdots{\,}O_{k_{p}}f ~=~ O_{k_{1}}\cdots{\,}O_{k_{i-1}}O_{k_{i+1}}O_{k_{i}}O_{k_{i+2}}\cdots{\,}O_{k_{p}}f.
$$
Let $g_{1}$ and $g_{2}$ be the $(n-p)$-ary functions obtained from the left-hand side and from the right-hand side of this equality by deleting
their inessential variables $x_{k_{1}},\ldots ,x_{k_{p}}$. If $O_{k_{i}}=\wedge_{k_{i}},O_{k_{i+1}}=\wedge _{k_{i+1}}$ or
$O_{k_{i}}=\vee_{k_{i}},O_{k_{i+1}}=\vee _{k_{i+1}}$, then $g_{1}=g_{2}$ holds trivially. Let us now assume that
$O_{k_{i}}=\vee_{k_{i}},O_{k_{i+1}}=\wedge _{k_{i+1}} $; the remaining case $O_{k_{i}}=\wedge _{k_{i}},O_{k_{i+1}}=\vee _{k_{i+1}}$ is similar.

By Proposition~\ref{prop basic properties of lattice derivatives}, we have $g_{1}\leq g_{2}$. Since the two (types of) functions given in
Theorem~\ref{thm reconstruction from derivatives} are order-incomparable, if $g_{1}\neq g_{2}$, then the lattice derivatives of $g_{1}$ and
$g_{2}$ cannot all coincide. Thus there exists $j\in [n]\setminus \{k_{1},\ldots ,k_{p}\}$ and $O_{j}\in \left\{ \wedge_{j},\vee_{j}\right\}$
such that $O_{j}g_{1}\neq O_{j}g_{2}$. Taking into account the definition of $g_{1}$ and $g_{2}$, we can rewrite this inequality as
$$
O_{j}O_{k_{1}}\cdots{\,}O_{k_{i-1}}O_{k_{i}}O_{k_{i+1}}O_{k_{i+2}}\cdots{\,}O_{k_{p}}f ~\neq ~O_{j}O_{k_{1}}\cdots{\,}O_{k_{i-1}}O_{k_{i+1}}O_{k_{i}}O_{k_{i+2}}\cdots{\,}O_{k_{p}}f,
$$
which contradicts the fact that $f$ has $(p+1)$-permutable lattice derivatives.
\end{proof}

If $f\colon \mathbb{B}^{n}\rightarrow \mathbb{B}$ is a Boolean function with $p$-permutable lattice derivatives for some $p\geq 2$, then $f$ has
$2$-permutable lattice derivatives by Theorem~\ref{thm (p+1)-permutable ==> p-permutable}, and then Theorem~\ref{thm Boolean 2-local iff 2-permutable} implies that $f$ is
$2$-locally monotone. Unfortunately, nothing more can be said about the degree of local monotonicity of a Boolean function with $p$-permutable
lattice derivatives. Indeed, the next example shows that there exist $n$-ary Boolean functions with $n$-permutable lattice derivatives that are
exactly $2$-locally monotone.

\begin{example}\label{ex n-permutable but only 3-local}
Let $f_{n}\colon \mathbb{B}^{n}\rightarrow \mathbb{B}$ be the function that takes the value $1$ on all tuples of the form
$$
\mathbf{x} ~=~ (\overset{m}{\overbrace{1,\ldots ,1}},0,\ldots ,0)~\text{ with } 0\leq m\leq n,
$$
and takes the value $0$ everywhere else. Using Corollary~\ref{cor Boolean 2-local via sections}, it is not difficult to verify that $f_{n}$ is
$2$-locally monotone. However, if $n\geq 3$, then $f_{n}$ is not $3$-locally monotone, since
\begin{eqnarray*}
\Delta _{2}f( 0,0,0,0,\ldots ,0) &=&-1, \\
\Delta _{2}f( 1,0,1,0,\ldots ,0) &=&\hphantom{-}1.
\end{eqnarray*}
Thus $f_{n}$ is exactly $2$-locally monotone.

We will show by induction on $n$ that $f_{n}$ has $n$-permutable lattice
derivatives. First we compute the meet derivatives%
$$
\wedge _{k}f_{n}(\mathbf{x})=\left\{ \!\!
\begin{array}{cl}
1{\,}, & \text{if }x_{1} = \cdots =x_{k-1}=1\text{ and }x_{k+1}=\cdots =x_{n}=0{\,};
\\
0{\,}, & \text{otherwise}{\,}.
\end{array}
\right.
$$
Since $\wedge_{k}f$ takes the value $1$ only at one tuple, it is monotone. The join derivative $\vee _{k}f_{n}$ is essentially the same as the
function $f_{n-1}$ (up to the inessential $k$-th variable of $\vee_{k}f_{n}$), that is,
\begin{equation}\label{eq join_k of  f_n}
\vee _{k}f_{n}(\mathbf{x}) ~=~ f_{n-1}(x_{1},\ldots ,x_{k-1},x_{k+1},\ldots ,x_{n}).
\end{equation}

Now it follows that if $\left\{ k_{1},\ldots ,k_{n}\right\} =[n]$ and $O_{k_{i}}\in \left\{ \wedge _{k_{i}},\vee _{k_{i}}\right\} $~$
(i=1,\ldots ,n)$, then
\begin{equation}\label{eq S_n-1}
O_{k_{1}}\cdots{\,}O_{k_{n-1}}O_{k_{n}}f ~=~ O_{k_{\pi (1)}}\cdots{\,}O_{k_{\pi (n-1)}}O_{k_{n}}f
\end{equation}
holds for every permutation $\pi \in S_{n-1}$. (If $O_{k_{n}}=\wedge _{k_{n}}$, then we use Theorem~\ref{thm p-local ==> p-permutable} and the
fact that $\wedge _{k_{n}}f$ is monotone, and if $O_{k_{n}}=\vee _{k_{n}}$, then we use (\ref{eq join_k of f_n}) and the induction hypothesis.)
On the other hand, from the $2$-local monotonicity of $f$ we can conclude that
\begin{equation}\label{eq (n-1 n)}
O_{k_{1}}\cdots{\,}O_{k_{n-2}}O_{k_{n-1}}O_{k_{n}}f ~=~ O_{k_{1}}\cdots{\,}O_{k_{n-2}}O_{k_{n}}O_{k_{n-1}}f
\end{equation}
with the help of Theorem~\ref{thm Boolean 2-local iff 2-permutable}. Since $S_{n}$ is generated by $S_{n-1}$ and the transposition $(n-1~n)$, we
see from (\ref{eq S_n-1}) and (\ref{eq (n-1 n)}) that $f$ has $n$-permutable lattice derivatives.
\end{example}

We finish this section with game-theoretic interpretations of the parameterized notions of local
monotonicity and permutability of lattice derivatives.
Identifying $\mathbb{B}^{n}$ with the power set of $[n]$,
we can regard a pseudo-Boolean function $f\colon\mathbb{B}%
^{n}\rightarrow\mathbb{R}$ as a cooperative game, where  $[n]$ is the set of players and
$f\left( C\right) $ is the worth of coalition $C\subseteq [n]$.

The partial derivative $\Delta_{k}f(C)$ gives the (marginal) contribution of the $k$-th
player to coalition $C$. Note that the same player might have a
positive contribution to some coalitions and a negative contribution to other
coalitions. Such a setting can model situations where some players have conflicts, which prevents them
from cooperating. The lattice derivative $\vee_{k}f(C)$
gives the outcome if the $k$-th player acts benevolently and
joins (or leaves) the coalition $C$ only if this increases the worth.
Similarly, $\wedge_{k}f(C)$ represents the outcome if the $k$-th player acts
malevolently.

Games corresponding to locally monotone functions have the property that if
two coalitions are close to each other, then any given player relates in the
same way to these coalitions. More precisely, $f$ is $p$-locally monotone if and only if
whenever two coalitions differ in less than $p$ players, then the contribution of any player is either nonnegative to both coalitions or
it is nonpositive to both.

Finally, let us interpret permutability of lattice derivatives. Let $P$ be a
$p$-subset of $[n]$, and let $C\subseteq [n]\setminus P$. Suppose that
some players of $P$ are benevolent and some of them are malevolent, and
they are asked one by one to join coalition $C$ if they want to. We obtain
the least possible outcome if the malevolent players are asked first, and we
get the greatest outcome if the benevolent players make their choices first.
The function $f$ has $p$-permutable lattice derivatives if and only if these
extremal outcomes coincide, i.e., if the order in which the players make their
choices is irrelevant for every $p$-subset of $[n]$.

%---------------------------------------------------------------------------------------------- Section 4
\section{Symmetric functions}\label{sect symmetric}

In the previous sections we saw that both notions of local monotonicity and of permutable lattice derivatives lead to two hierarchies of
pseudo-Boolean functions which are related by the fact that each $p$-local monotone class is contained in the corresponding class of functions
having $p$-permutable lattice derivatives. Now, in general this containment is strict. However, under certain assumptions (see, e.g., Theorem~\ref{thm
Boolean 2-local iff 2-permutable}), $p$-local monotonicity is equivalent to $p$-permutability of lattice derivatives. Hence it is natural to ask
for conditions under which these two notions are equivalent.

In this section we provide a partial answer to this problem by focusing on symmetric pseudo-Boolean functions, i.e., functions $f\colon
\mathbb{B}^{n}\rightarrow \mathbb{R}$ that are invariant under all permutations of their variables. Quite surprisingly, in this case the notions
of $p$-local monotonicity and $p$-permutability of lattice derivatives become equivalent.

Symmetric functions of arity $n$ are in a one-to-one correspondence with sequences of real numbers of length $n+1$, where the function
corresponding to the sequence $\alpha =\alpha _{0},\ldots ,\alpha _{n}$ is given by $f(\mathbf{x})=\alpha _{\left\vert \mathbf{x}\right\vert
}~(\mathbf{x}\in \mathbb{B}^{n})$. Clearly, $f$ is isotone if and only if the corresponding sequence is nondecreasing, i.e., $\alpha _{0}\leq
\alpha _{1}\leq \cdots \leq \alpha _{n}$. Similarly, $f$ is antitone if and only if $\alpha _{0}\geq \alpha _{1}\geq \cdots \geq \alpha _{n}$,
and $f$ is monotone if and only if $f$ is either isotone or antitone.\footnote{Since if $f$ is isotone (resp. antitone) in one variable, then it
is isotone (resp. antitone) in all variables.}

It is easy to see that if $f$ is symmetric, then every section of $f$ is also symmetric; moreover, if $f$ corresponds to the sequence $\alpha
=\alpha _{0},\ldots ,\alpha _{n}$, then the $p$-ary sections of $f$ are precisely the symmetric functions corresponding to the
subsequences\footnote{Here by a subsequence we mean a sequence of consecutive entries of the original sequence.} $\alpha _{i},\alpha
_{i+1},\ldots ,\alpha_{i+p}$ of $\alpha $ of length $p+1$. This observation and  Theorem~\ref{thm p-local via sections} lead to the following
description of $p$-locally monotone symmetric pseudo-Boolean functions.

\begin{proposition}\label{prop symmetric p-local via sections}
Let $f\colon \mathbb{B} ^{n}\rightarrow \mathbb{R}$ be a symmetric function corresponding to the sequence $\alpha =\alpha _{0},\ldots
,\alpha_{n}$. Then $f$ is $p$-locally monotone if and only if each subsequence of length $p+1$ of $\alpha $ is either nondecreasing or
nonincreasing.
\end{proposition}

Unlike in the previous sections, here it will be more convenient to discard the inessential $k$-th variable of the lattice derivatives
$\wedge_{k}f$ and $\vee_{k}f$, and regard the latter as $(n-1)$-ary functions. Clearly, if $f$ is symmetric, then so are its lattice
derivatives. Moreover, if $f$ corresponds to the sequence $\alpha =\alpha_{0},\ldots ,\alpha_{n}$, then $\wedge_{k}f$ and $\vee_{k}f$ correspond
to the sequences
\begin{eqnarray*}
&&\alpha _{0}\wedge \alpha _{1},\alpha _{1}\wedge \alpha _{2},\ldots ,\alpha
_{n-1}\wedge \alpha _{n}~\text{ and} \\
&&\alpha _{0}\vee \alpha _{1},\alpha _{1}\vee \alpha _{2},\ldots ,\alpha _{n-1}\vee \alpha _{n}{\,},
\end{eqnarray*}
respectively, for all $k\in [n]$. Since these sequences do not depend on $k$, we will write $\wedge f$ and $\vee f$ instead of $\wedge _{k}f$
and $\vee_{k}f$, and we will abbreviate
\[
 \underbrace{\wedge \cdots \wedge}_{\ell} f\quad \text{ and } \quad \underbrace{\vee \cdots \vee}_{\ell} f
\]
by $\wedge^{\ell }f$ and $\vee^{\ell }f$, respectively.

The next example shows that Theorem~\ref{thm (p-1)-locally monotone derivatives} cannot be sharpened.

\begin{example}\label{ex f p-local Vf (p-1)-local}
Let $f\colon \mathbb{B}^{n}\rightarrow \mathbb{B}$ be the symmetric function corresponding to the sequence
$$
\alpha  ~=~ 0,0,\overset{p}{\overbrace{1,\ldots ,1}},\overset{p}{\overbrace{0,\ldots ,0}},1,1
$$
where $n=2p+4$ and $p\geq 2$. It follows from Proposition~\ref{prop symmetric p-local via sections} that $f$ is exactly $p$-locally monotone. To
compute $\wedge f$, it is handy to construct a table whose first row contains the sequence $\alpha $, and in the second row we write $\alpha
_{i}\wedge \alpha _{i+1}$ between $\alpha _{i}$ and $\alpha _{i+1}$:
$$
\begin{tabular}{rc@{}c@{}c@{}c@{}c@{}c@{}c@{}c@{}c@{}c@{}c@{}c@{}c@{}c@{}c@{}c@{}c@{}c@{}c@{}c@{}c@{}c@{}c@{}c@{}c@{}c@{}c@{}c@{}c@{}c@{}c}
$f:$ & $0$ &  & $0$ &  & $1$ &  & $1$ &  & $\cdots $ &  & $1$ &  & $1$ &  & $%
0$ &  & $0$ &  & $\cdots $ &  & $0$ &  & $0$ &  & $1$ &  & $1$ &  &  &  &
\\
$\wedge f:$ &  & $0$ &  & $0$ &  & $1$ &  & $1$ & $~\cdots \,~$ & $1$ &  & $%
1 $ &  & $0$ &  & $0$ &  & $0$ & $~\cdots ~$ & $0$ &  & $0$ &  & $0$ &  & $1$ &  &  &  &  &
\end{tabular}
$$
Thus $\wedge f$ corresponds to the sequence
$$
0,0,\overset{p-1}{\overbrace{1,\ldots ,1}},\overset{p+1}{\overbrace{0,\ldots ,0}},1,
$$
and a similar calculation yields that $\vee f$ corresponds to the sequence
$$
0,\overset{p+1}{\overbrace{1,\ldots ,1}},\overset{p-1}{\overbrace{0,\ldots ,0}},1,1.
$$
Now Proposition~\ref{prop symmetric p-local via sections} shows that $\wedge f$ and $\vee f$ are exactly $(p-1)$-locally monotone.
\end{example}

\begin{remark}
Example~\ref{ex f p-local Vf (p-1)-local} shows that the degree of local monotonicity can decrease, when taking lattice derivatives, and Theorem~\ref{thm (p-1)-locally
monotone derivatives} states that it can decrease by at most one. Other examples can be found to illustrate the cases when this degree stays the
same, or even increases. For instance, consider the function $f(\mathbf{x})=x_{1}\oplus \cdots \oplus x_{n}$, which is not even $2$-locally
monotone, but its lattice derivatives are constant.
\end{remark}

We conclude this section by proving that for symmetric functions the notions of $p$-local monotonicity and $p$-permutability of lattice
derivatives coincide.

\begin{theorem}\label{thm symmetric p-local iff p-permutable}
If $f\colon \mathbb{B}^{n}\rightarrow \mathbb{R}$ is symmetric, then $f$ is $p$-locally monotone if and only if $f$ has
$p$-permutable lattice derivatives.
\end{theorem}

\begin{proof}
By Theorem~\ref{thm p-local ==> p-permutable}, it is enough to show that if a symmetric function $f\colon \mathbb{B}^{n}\rightarrow \mathbb{R}$
is not $p$-locally monotone, then it does not have $p$-permutable lattice derivatives. So suppose that $f$ is a symmetric function which is not
$p$-locally monotone, and which corresponds to the sequence $\alpha =\alpha_{0},\ldots ,\alpha_{n}$. Let $\alpha_{i},\dots ,\alpha _{i+\ell}$ be
a shortest subsequence of $\alpha$ that is neither nondecreasing nor nonincreasing. Proposition~\ref{prop symmetric p-local via sections}
implies that there is indeed such a subsequence for $\ell +1\leq p+1$. From the minimality of $\ell $ it follows that the subsequence
$\alpha_{i},\dots ,\alpha_{i+\ell -1}$ is either nondecreasing or nonincreasing. We may assume without loss of generality that the first case
holds; the second case is the dual of the first one. Then we must have $\alpha_{i+\ell -1}>\alpha_{i+\ell}$, since otherwise the whole
subsequence $\alpha_{i},\dots ,\alpha_{i+\ell }$ would be nondecreasing. Thus we have the following inequalities:
\begin{equation}\label{eq alpha nondecreasing}
\alpha _{i}~\leq ~\alpha_{i+1}~\leq ~\cdots ~\leq ~\alpha_{i+\ell -1} ~>~\alpha _{i+\ell }{\,}.
\end{equation}
From the minimality of $\ell$, we can also conclude that the subsequence $ \alpha _{i+1},\dots ,\alpha _{i+\ell}$ is either nondecreasing or
nonincreasing. As $\alpha _{i+\ell -1}>\alpha _{i+\ell}$, the first case is impossible, therefore $\alpha _{i+1},\dots ,\alpha _{i+\ell }$ is
nonincreasing, and we must have $\alpha _{i}<\alpha _{i+1}$ since otherwise the whole subsequence $\alpha _{i},\dots ,\alpha _{i+\ell }$ would
be nonincreasing:
\begin{equation}\label{eq alpha nonincreasing}
\alpha_{i} ~<~\alpha_{i+1}~\geq ~\cdots ~\geq ~\alpha_{i+\ell -1}~\geq ~\alpha_{i+\ell }{\,}.
\end{equation}

Comparing (\ref{eq alpha nondecreasing}) and (\ref{eq alpha nonincreasing}), we obtain
$$
\alpha_{i}~<~\alpha_{i+1} ~=~\cdots ~=~\alpha_{i+\ell -1}~>~\alpha_{i+\ell }{\,}.
$$
To simplify notation, we set $\beta :=\alpha _{i},\gamma :=\alpha _{i+1},\delta :=\alpha _{i+\ell }$. With this notation we have that $\alpha$
contains the subsequence $\beta ,\gamma ,\ldots ,\gamma ,\delta $ of length $\ell +1$ with $\beta ,\delta <\gamma $. In the following we will
use this observation to prove that $f$ does not have $\ell $-permutable lattice derivatives.

Let us compute the sequence corresponding to $\vee \wedge^{\ell -1}f$. We can construct a table as in Example~\ref{ex f p-local Vf (p-1)-local},
but this time the table has $\ell +1$ rows (in the last row $\mu $ stands for $\beta \vee \delta $):
$$
\begin{tabular}{rp{0.4cm}@{}p{0.4cm}@{}p{0.4cm}@{}p{0.4cm}@{}p{0.4cm}@{}p{0.4cm}@{}p{0.4cm}@{}p{0.4cm}@{}p{0.4cm}@{}p{0.4cm}@{}p{0.4cm}@{}p{0.4cm}@{}p{0.4cm}@{}p{0.4cm}@{}p{0.4cm}@{}p{0.4cm}@{}p{0.4cm}@{}p{0.4cm}@{}p{0.4cm}@{}p{0.4cm}@{}p{0.4cm}@{}p{0.4cm}@{}p{0.4cm}@{}p{0.4cm}@{}p{0.4cm}}
$f:$ & $\alpha _{0}$ &  & $\cdots $ &  & $\beta $ &  & $\gamma $ &  & $%
\gamma $ &  & $\gamma $ &  & $\cdots $ &  & $\gamma $ &  & $\gamma $ &  & $%
\gamma $ &  & $\delta $ &  & $\cdots $ &  & $\alpha _{n}$ \\
$\wedge f:$ &  &  &  & $\cdots $ &  & $\beta $ &  & $\gamma $ &  & $\gamma $
&  &  & $\cdots $ &  &  & $\gamma $ &  & $\gamma $ &  & $\delta $ &  & $%
\cdots $ &  &  &  \\
$\wedge ^{2}f:$ &  &  &  &  & $\cdots $ &  & $\beta $ &  & $\gamma $ &  &  & & $\cdots $ &  &  &  & $\gamma $ &  & $\delta $ &  & $\cdots $ &  &
&  &
\\
&  &  &  &  &  &  & $\cdots $ &  & $\cdots $ &  & $\cdots $ &  & $\cdots $ &
& $\cdots $ &  & $\cdots $ &  & $\cdots $ &  &  &  &  &  &  \\
$\wedge ^{\ell -2}f:$ &  &  &  &  &  &  &  &  & $\cdots $ &  & $\beta $ &  &
$\gamma $ &  & $\delta $ &  & $\cdots $ &  &  &  &  &  &  &  &  \\
$\wedge ^{\ell -1}f:$ &  &  &  &  &  &  &  &  &  & $\cdots $ &  & $\beta $ &
& $\delta $ &  & $\cdots $ &  &  &  &  &  &  &  &  &  \\
$\vee \wedge ^{\ell -1}f:$ &  &  &  &  &  &  &  &  &  &  & $\cdots $ &  & $%
\mu $ &  & $\cdots $ &  &  &  &  &  &  &  &  &  &
\end{tabular}
$$
A similar table can be constructed for $\wedge^{\ell -1}\vee f$:
$$
\begin{tabular}{rp{0.4cm}@{}p{0.4cm}@{}p{0.4cm}@{}p{0.4cm}@{}p{0.4cm}@{}p{0.4cm}@{}p{0.4cm}@{}p{0.4cm}@{}p{0.4cm}@{}p{0.4cm}@{}p{0.4cm}@{}p{0.4cm}@{}p{0.4cm}@{}p{0.4cm}@{}p{0.4cm}@{}p{0.4cm}@{}p{0.4cm}@{}p{0.4cm}@{}p{0.4cm}@{}p{0.4cm}@{}p{0.4cm}@{}p{0.4cm}@{}p{0.4cm}@{}p{0.4cm}@{}p{0.4cm}}
$f:$ & $\alpha _{0}$ &  & $\cdots $ &  & $\beta $ &  & $\gamma $ &  & $%
\gamma $ &  & $\gamma $ &  & $\cdots $ &  & $\gamma $ &  & $\gamma $ &  & $%
\gamma $ &  & $\delta $ &  & $\cdots $ &  & $\alpha _{n}$ \\
$\vee f:$ &  &  &  & $\cdots $ &  & $\gamma $ &  & $\gamma $ &  & $\gamma $
&  &  & $\cdots $ &  &  & $\gamma $ &  & $\gamma $ &  & $\gamma $ &  & $%
\cdots $ &  &  &  \\
$\wedge \vee f:$ &  &  &  &  & $\cdots $ &  & $\gamma $ &  & $\gamma $ &  & &  & $\cdots $ &  &  &  & $\gamma $ &  & $\gamma $ &  & $\cdots $ &
&  &  &
\\
&  &  &  &  &  &  & $\cdots $ &  & $\cdots $ &  & $\cdots $ &  & $\cdots $ &
& $\cdots $ &  & $\cdots $ &  & $\cdots $ &  &  &  &  &  &  \\
$\wedge ^{\ell -3}\vee f:$ &  &  &  &  &  &  &  &  & $\cdots $ &  & $\gamma $
&  & $\gamma $ &  & $\gamma $ &  & $\cdots $ &  &  &  &  &  &  &  &  \\
$\wedge ^{\ell -2}\vee f:$ &  &  &  &  &  &  &  &  &  & $\cdots $ &  & $%
\gamma $ &  & $\gamma $ &  & $\cdots $ &  &  &  &  &  &  &  &  &  \\
$\wedge ^{\ell -1}\vee f:$ &  &  &  &  &  &  &  &  &  &  & $\cdots $ &  & $%
\gamma $ &  & $\cdots $ &  &  &  &  &  &  &  &  &  &
\end{tabular}
$$
Since $\beta ,\delta <\gamma $, we have $\mu <\gamma $, and this means that the sequences corresponding to $\vee \wedge ^{\ell -1}f$ and $\wedge
^{\ell -1}\vee f$ differ in at least one position, therefore $f$ does not have $\ell $-permutable lattice derivatives. As $\ell \leq p$, this
implies that $f $ does not have $p$-permutable lattice derivatives either, according to Theorem~\ref{thm (p+1)-permutable ==> p-permutable}.
\end{proof}

\begin{remark}\label{rem symmetric exactly p-permutable}
As a consequence of Theorem~\ref{thm symmetric p-local iff p-permutable}, we can observe that any exactly $p$-locally monotone symmetric
function (for instance, the functions considered in Example~\ref{ex f p-local Vf (p-1)-local}) has $p$-permutable but not $(p+1)$-permutable
lattice derivatives.
\end{remark}

%---------------------------------------------------------------------------------------------- Section 5
\section{Open problems and concluding remarks}

We proposed relaxations of monotonicity, namely $p$-local monotonicity, and we presented characterizations of each in terms of ``forbidden''
sections. Also, for each $p$, we observed that $p$-locally monotone functions have the property that any $p$ of their lattice derivatives permute,
and showed that the converse also holds in the special case of symmetric functions. The classes of $2$-locally monotone functions, and of
functions having $2$-permutable lattice derivatives were explicitly described. However,
similar descriptions elude us for $p\geq 3$. Hence we are left with the following problems.

\begin{problem}
For $p\geq 3$, describe the class of $p$-locally monotone functions and that of functions having $p$-permutable lattice derivatives.
\end{problem}

\begin{problem}
For $p\geq 3$, determine necessary and sufficient conditions on functions for the equivalence between $p$-local monotonicity and $p$-permutability of lattice derivatives.
\end{problem}

%---------------------------------------------------------------------------------------------- Acknowledgments
\section*{Acknowledgments}

Miguel Couceiro and Jean-Luc Marichal are supported by the internal research project F1R-MTH-PUL-12RDO2 of the University of Luxembourg.
Tam\'as Waldhauser is supported by the \hbox{T\'{A}MOP-4.2.1/B-09/1/KONV-2010-0005} program of National
Development Agency of Hungary, by the Hungarian National Foundation for
Scientific Research under grant nos K77409 and K83219, by the National
Research Fund of Luxembourg, and by the Marie Curie Actions of
the European Commission \hbox{(FP7-COFUND).}

\end{document}